\documentclass[a4paper, 11 pt]{article}
\usepackage [a4paper,left=3.1 cm,bottom=2.7cm,right=3.1 cm,top=2.7cm]{geometry}
\usepackage[english]{babel}
\usepackage{amssymb}
\usepackage{amsmath}
\usepackage{amsthm}
\usepackage{mathdesign}
\usepackage{pdfsync}
\usepackage{color}
\usepackage[colorlinks]{hyperref}

\usepackage{csquotes}
\usepackage{subcaption}
\usepackage{tabularx,array}
\usepackage{graphicx}
\usepackage{enumitem} 
\usepackage{mathtools}
\usepackage{xcolor}
\def\E{\mathbb{E}}

\newtheorem{remark}{Remark}
\newtheorem{theorem}{Theorem}
\newtheorem{lemma}{Lemma}

\newtheorem{proposition}{Proposition}

\newtheorem{h}{Hypothesis}

\newcommand{\keywords}[1]{\par\addvspace\baselineskip
\noindent\enspace\ignorespaces#1}

\begin{document}

\title{On the implied volatility of European and Asian call options under the stochastic volatility Bachelier model}

\author{Elisa Al\`os$^{\ast}$, Eulalia Nualart$^{\ast}$, and Makar Pravosud\thanks{ 
Universitat Pompeu Fabra and Barcelona School of Economics, Department of Economics and Business, Ram\'on Trias Fargas 25-27, 08005, Barcelona, Spain. 
EN acknowledges support from the Spanish MINECO grant PID2022-138268NB-100 and
Ayudas Fundacion BBVA a Equipos de Investigaci\'on Cient\'ifica 2021.}}
\maketitle

\begin{abstract}
 In this paper we study the short-time behavior of the at-the-money implied volatility for European and arithmetic Asian call options with fixed strike price.
 The asset price is assumed to follow the Bachelier model with a general stochastic volatility process. Using techniques of the Malliavin calculus such as the anticipating It\^o's formula we first compute the level of the implied volatility when the maturity converges to zero. Then, we find a short maturity asymptotic  formula for the skew of the implied volatility that depends on the roughness of the volatility model. We apply our general results to the SABR, fractional Bergomi and local volatility models, and provide some numerical simulations that confirm the accurateness of the asymptotic formula for the skew.
\end{abstract}


\section{Introduction}\label{introduction}
The life of quantitative finance and derivatives pricing has started when Louis Bachelier presented an option pricing model in his Ph.D. thesis \cite{ASENS_1900_3_17__21_0}. The
Bachelier model assumes that asset prices follow a Brownian motion, a continuous-time stochastic process with normally distributed increments. Then, price changes are normally distributed with a constant volatility and the model allows for negative prices. Due to the possibility of negative prices, the Bachelier model has not gained popularity. Nowadays option pricing models are strongly based on the Black-Scholes model. Local,  stochastic and even fractional volatility models (see \cite{Alos2021c} for an introduction to these topics) are extensions of the Black-Scholes model, where the volatility is allowed to be a function of the price of the stock price (local), a diffusion process (stochastic) or a function of a fractional Brownian motion (fractional). As the Black-Scholes model is a geometric Brownian motion (following a log-normal distribution), all the above models are also exponentials of random variables, and they share the positivity of the stock price. This plausible and intuitive property has been assumed for a long time for the prices of the underlying asset of an option.

From a practical perspective, the Bachelier model has started to gain more attention when negative prices have been registered, for example, on April 20th 2020, when crude oil futures crossed the zero mark and took negative values. Moreover, in some commodities markets, the Bachelier model is more adequate to modeling price volatility, as commodity prices do not always follow the same multiplicative dynamics observed in the stock market. As a prompt response to enable markets to continue to function normally, the CME Clearing started to use the Bachelier model to cope with negative underlying prices.  The empirical paper by \cite{Galeeva2022} studies the turbulent period on the oil-futures options markets which was observed in the middle of February 2020. The authors analyzed the observed implied volatilities by applying the Bachelier and Black models to prevail oil-futures option prices.

After the 2008 financial crisis and the adoption of negative interest rate policies by central banks in Europe and Japan, the Bachelier model became a useful tool, since it allows interest rates to approach or fall below zero. In particular, this model becomes interesting in the swaption market (options on interest rate swaps). A detailed analysis on the reasons to use a normal model rather than a log-normal one in the case of interest rates is presented in \cite{HoGoodmanLaurieS2003}. The recent paper by \cite{https://doi.org/10.1002/fut.22315} is a comprehensive review of various topics related to the Bachelier model for both researchers and practitioners. In particular, they cover topics such as the implied volatility inversion problem, volatility conversion between related models, Greeks and hedging, as well as pricing of exotic options.

Concerning exact analytical results, in the paper by \cite{https://doi.org/10.1111/j.1467-9965.2007.00326.x}, the authors compare the option pricing problem under the Bachelier and Black-Scholes models. They find that prices coincide very well and explain by the means of chaos expansion theory why the Bachelier model yields good short-time approximations of prices and volatilities. In addition to that, the computation of prices and implied volatilities under the Bachelier model model has been presented in, for example, \cite{Terakado2019}. However, the existing literature lacks the rigorous analysis of the behaviour of the ATM level and skew of the implied volatility under the Bachelier model with stochastic volatility. By contrast, in the case of a regular Black-Scholes framework, there are numerous papers that study the short end maturity behaviour of the implied volatility under a variety of models such as the  Heston, the SABR and fractional Begomi. For example, see \cite{Alos2007a}, \cite{Figueroa-Lopez2016},   \cite{doi:10.1080/14697688.2016.1197410}, \cite{fouque}, \cite{Pirjol2016}, \cite{euch}.

The goal of this paper is to fill the gap in the existing literature and to analyze the option pricing problem when we allow the volatility of the Bachelier model to be a stochastic process. That is, to study option prices and implied volatilities for local, stochastic and fractional Bachelier models, both for vanilla and for exotic, in particular, Asian options. Specifically, we present analytical results for the behaviour of the level and the skew of the implied volatility for a general stochastic volatility Bachelier model. Our main tool for proving these results is Malliavin calculus, see Appendix \ref{MCintro} for an introduction to this topic.
We then apply our asymptotic formulas to the SABR, fractional Bergomi and local volatility models. Additionally, we compare our first order approximations of the implied volatility with the classical Monte Carlo method across different moneyness and maturity . The presented results are extension of the papers  \cite{Alos2006}, \cite{Alos2018} and \cite{mp2022}, where the Malliavin calculus is also used in order to obtain short maturity asymptotic for the level and skew of European and  Asian options under the Black and Scholes model with stochastic volatility.

The main advantage of the Malliavin calculus approach is its generality and simplicity. The results in this paper lead to straightforward expressions in terms of Malliavin derivatives. These derivatives can be easily and explicitly computed for most volatility models (see, for example, \cite{Alos2021c}), resulting in simple algebraic expressions which are uniform over all the model parameters. Thus, our results provide general formulas that can be directly applied to different scenarios by substituting the volatility derivatives with their corresponding explicit expressions in each case (see Section 5). There exist other techniques in order to obtain short-maturity asymptotic expansions of the implied volatility. See for example \cite{euch} and the references therein. In this paper, the authors obtain a small-time Edgeworth expansion of the density of an asset price under a general stochastic volatility models, from which expansions of put option prices and at-the-money implied volatility follow. However, it is not clear if those methods extend in the case of exotic options such as Asian options. For local volatility models, a Taylor expansion for short maturity asymptotics for Asian options is obtained in \cite{Pirjol2016}.

The paper is organized as follows. Section 2 is devoted to the statement of the problem and main results. Intermediary steps and the proofs of the main results are presented in Sections 3 and 4. Finally, in Section 5 we present and discuss the results of the numerical study.

\section{Statement of the problem and main results}\label{statement-of-the-problem-and-notation}

Let  $T>0$ and consider the following model for asset prices $S$ (without lost of generality, we take the interest rate equal to zero for the sake of simplicity) in a time interval $[0,T]$
\begin{equation} \begin{split}\label{B_Epm}
dS_t &= \sigma_tdW_t\\
 W_t &= \rho W_t' + \sqrt{(1-\rho^2)}B_t,
 \end{split}
\end{equation} where $S_0>0$ is fixed, $W_t$, $W_t'$, and $B_t$ are three standard Brownian motions on $[0,T]$
defined on the same risk-neutral complete probability space $(\Omega, \mathcal{G}, \mathbb{P})$. We assume that \(W_t'\) and
\(B_t\) are independent and \(\rho\in (-1,1)\) is the correlation coefficient between \(W_t\) and 
\
\(W_t'\), and that $\sigma_t$ is an a.s. continuous and square integrable process adapted to the filtration generated by $W'_t$. We denote by $\E_t$ the conditional expectation with respect to the filtration generated by $W_t$.

Observe that when the volatility $\sigma$ is constant, equation (\ref{B_Epm}) is the classical Bachelier model. Notice that, in this case, asset prices are normally distributed.

We consider the following hypotheses on the volatility of the asset price. 
\begin{h} \label{Hyp1}
There exist positive constants $c_1$ and $c_2$ such that for all $t \in [0,T]$,
$$
c_1 \leq \sigma_t \leq c_2.
$$
\end{h}

%\begin{h}\label{Hyp2}
%For all $p\geq 1$ there exists $c>0$ and  $\gamma >0$ such that  for all $0\leq s\leq r\leq T \leq 1$,
%$$(\mathbb E|\sigma_r-\sigma_s|^p)^{1/p} \leq  c (r-s)^{\gamma}.$$
%\end{h}

\begin{h}\label{Hyp3}
$\sigma \in \mathbb{L}^{2,2}_{W'}$.
\end{h}

\begin{h}\label{Hyp4}
There exists  $H\in (0,1)$ and constants $c_1,c_2>0$ such  that for $0  \leq s \leq r \leq u \leq T \leq 1$ a.e.
\begin{equation} \label{d1}
\{\E (\vert D_r^{W'}\sigma_u \vert^2)\}^{1/2}\leq c_1 (u-r)^{H-\frac{1}{2}}
\end{equation}
and
\begin{equation} \label{d2}
\{\E (\vert D_s^{W'} D_r^{W'}\sigma_u \vert^2)\}^{1/2}\leq c_2 (u-r)^{H-\frac{1}{2}} (u-s)^{H-\frac{1}{2}}.
\end{equation}
\end{h}

\begin{remark}
These hypotheses have been taken for the sake of simplicity and could be replaced by more general integrability conditions as is done for instance in \cite{euch}. The arguments in this paper are general (in the sense that they are not model-specific) and can be technically adapted to deal with different volatility models with a simple truncation argument (see Section 5).
\end{remark}

We denote by  $(V^E_t)_{t \in [0,T]}$ the value of a European call option with fixed strike $k$. In particular, $$V^E_T=\E(S_T-k)_{+}.$$  In a similar way, we denote by $(V^A_t)_{t \in [0,T]}$ the price of an arithmetic Asian call option with fixed strike $k$. Then
$$V^A_T=\E(A_T-k)_+,$$ where $A_T=\frac{1}{T}\int_0^T S_t dt$.

In oder to to deal with Asian options we follow the same approach as in \cite{Alos2018}, that is, we consider the  martingale $M_t=\E_t(A_T)$.  Applying the stochastic Fubini's theorem we get that
\begin{align*}
\begin{split}
  A_T=\frac{1}{T}\int_0^T S_t dt &= \frac{1}{T}\int_0^T \left(S_0+\int_0^t \sigma_u dW_u\right) dt = \\
   &= S_0+\frac{1}{T}\int_0^T\sigma_u \left(\int_u^T  dt\right) dW_u.
\end{split}
\end{align*}
This implies that 
\begin{equation}\label{dACE}
dM_t=\frac{\sigma_t (T-t)}{T} dW_t = \phi_t dW_t,
\end{equation} where $$\phi_t := \frac{\sigma_t (T-t)}{T}.$$

Notice that, under the classical Bachelier model, $A_T$ is a Gaussian random variable with mean $S_0$ and variance $\frac{\sigma^2T}{3}$. 

We denote by \(B_E(t,S_t,k,\sigma)\)  the classical Bachelier price of a European call option with time to maturity \(T-t\), current stock price \(S_t\),
strike price \(k\) and constant volatility \(\sigma\). That is, $$
  B_E(t,S_t,k,\sigma)=(S_t-k)N(d_E(k,\sigma))+ n(d_E(k,\sigma))\sigma\sqrt{T-t},$$
where
$$
  d_E(k,\sigma)=\frac{S_t-k}{\sigma\sqrt{T-t}}.$$ Here \(N\) and \(n\) denote the cumulative distribution
function and the probability density  function of a standard normal random variable, respectively. Additionally, we recall that the Bachelier price satisfies the
following PDE \begin{align}
  \partial_t B_E(t,x,k,\sigma) + \frac{1}{2}\sigma^2 \partial^2_{xx} B_E(t,x,k,\sigma) = 0.
\label{B_Epde}
\end{align}

Similarly, we denote by \(B_A(t,S_t,y_t,k,\sigma)\)  the Bachelier price of an arithmetic Asian call option with constant volatility \(\sigma\), where $y_t=\int_0^tS_udu$. A direct computation shows  that (see the Appendix   for the details)
$$
B_A(t,S_t,y_t,k,\sigma)=\left(S_t\frac{T-t}{T}+\frac{y_t}{T}-k\right)N(d_A(k,\sigma))+\left(\frac{\sigma(T-t)\sqrt{T-t}}{T\sqrt{3}}\right)n(d_A(k,\sigma)),
$$
where
$$
d_A(k,\sigma)=\frac{S_t\frac{T-t}{T}+\frac{y_t}{T}-k}{\left(\frac{\sigma(T-t)\sqrt{T-t}}{T\sqrt{3}}\right)}.
$$
 
Notice that we have the relation
$$
 B_A(t,S_t,y_t,k,\sigma)= B_E\left(t,M_t,k,\frac{\sigma(T-t)}{T\sqrt{3}}\right).
$$

We next define the implied volatility (IV)  of a European call option as the quantity $I_E(t,k)$ such that $$V_t^E=B_E(t,S_t,k,I_E(t,k)),$$ and we denote by \(I_E(t,k^{*}_t)\), where $k^*_t=S_t$,  the corresponding at-the-money implied volatility (ATMIV) which, in the case of zero interest rates, takes the form
\(B_E^{-1}(t,S_t,S_t,V_t^E)\). Similarly, we define the implied volatility of an Asian call option $I_A(t,k)$ as the quantity such that
$$V^A_t=B_A(t,S_t, y_t, k,I_A(t,k)),$$ and we denote by \(I_A(t,k_t^{*})\) the corresponding ATMIV which, in the case of zero interest rates, takes the form 
\(B_A^{-1}(t,S_t,y_t,S_t\frac{T-t}{T}+\frac{y_t}{T},V^A_t)\). We set $k^{\ast}=k^{\ast}_0$.

The aim of this paper is to apply  the Malliavin calculus techniques developed in \cite{Alos2006} and \cite{Alos2018} in order to obtain formulas for  the ATMIV level and skew as $T \rightarrow 0$ under the general stochastic volatility model (\ref{B_Epm}).

The main result of this paper is the following theorem.
\begin{theorem}  \label{limskew}
Assume Hypotheses \ref{Hyp1}-\ref{Hyp4}. Then, 
\begin{align} \label{main1}
\lim_{T\to 0}I_E(0,k^*)=\sigma_0 \quad \text{and} \quad \lim_{T\to 0}I_A(0,k^*)=\sigma_0.
\end{align}
Moreover,
\begin{equation} \label{main2}\begin{split}
&\lim_{T \to 0} T^{\max(\frac12-H, 0)}\partial_kI_E(0,k^{*}) \\
&=\lim_{T \to 0} T^{\max(\frac12-H, 0)}\frac{\rho}{ \sigma_0 T^2}\int_0^T \int_r^T\E(D_r^{W'}\sigma_u) du  dr\\
 \end{split}
\end{equation}
\text{and} 
\begin{equation} \label{main22}\begin{split}
&\lim_{T \to 0} T^{\max(\frac12-H, 0)}\partial_kI_A(0,k^{*})\\
&=\lim_{T \to 0} T^{\max(\frac12-H, 0)}\frac{9 \rho}{ \sigma_0 T^5}\int_0^T (T-r)\int_r^T(T-u)^2\E(D_r^{W'}\sigma_u) du  dr,
 \end{split}
\end{equation}
provided that all limits exist.
\end{theorem}

Observe that when prices and volatilities are uncorrelated  then the short-time skew equals to zero, which  coincides with the constant volatility  case. Notice also that if the term $\E(D_r^{W'}\sigma_u)$ is of order $(u-r)^{H-\frac12}$, the quantity multiplying the term 
$T^{\max(\frac12-H, 0)}$ of the right hand side of (\ref{main2}) and (\ref{main22}) is bounded by $cT^{H-\frac12}$. In particular, this limit is $0$ if $H>1/2$. This suggest that, in the case $H<1/2$,  we need to multiply by $T^{\frac12-H}$ in order to obtain a finite limit.
See the examples in Section 5.

The results of Theorem  \ref{limskew} can be used in order to derive approximation formulas for the price of European and Asian call options. Notice that
\begin{equation*} 
V^E_0= B_E(0,S_0,k,I_E(0,k)) \quad \text{and} \quad V^A_0= B_A(0,S_0,y_0,k,I_A(0,k)).
\end{equation*}
Then, using Taylor's formula we can use the approximations
\begin{equation} \label{proxyPrice}
\begin{split}
I_E(0,k) &\approx I_E(0,k^{*}) + \partial_kI_E(0,k^{*})(k-k^{*}),\\
I_A(0,k) &\approx I_A(0,k^{*}) + \partial_kI_A(0,k^{*})(k-k^{*}).
\end{split}
\end{equation}
The great utility of these relations is  that we can use them to approximate the price of European and Asian call options for a wide range of stochastic and fractional volatility models. We numerically compare this approximation with the classical Monte Carlo approximation for the SABR and fractional Bergomi models in Section \ref{proxyNumerics}. 

In conclusion, Theorem 1 presents a a short maturity asymptotic formula for both the level and skew of the ATMIV of both European and Asian call options under the Bachelier model of stock prices and general stochastic volatility models. This formulas only depend on the Malliavin derivative of the underlying volatility model and applies to a large class of stochastic, fractional and local volatility models. 
As an application, we compare this formula with the  linear approximation (\ref{proxyPrice}) and the classical Monte Carlo. Of course, would expect to obtain better results if one has a short maturity asymptotic formula for the curvature $\partial_{kk}^2I(0,k^{*})$. The short-time maturity asymptotics for the at-the-money curvature of the implied volatility for European calls under general stochastic volatility models is computed in \cite{Alos2017}. See also \cite{euch} for asymptotic expansions of put option prices and at-the-money implied volatility 
under general volatility models using small-time Edgeworth expansion of the density of an asset price. A Taylor expansion for short maturity asymptotics for Asian options when the underlying asset follows a local volatility model is obtained in \cite{Pirjol2016}.
In our setting, computing the curvature using the Malliavin calculus technique following \cite{Alos2017} would be the next step and we leave it for further work. Another possible extension of this paper would be to add jumps in the underlying process as in \cite{Alos2007a}.

\section{Preliminary results: decomposition formulas}

In this section we provide closed form decomposition formulas for the prices and for the ATM implied volatility skew of European and Asian call options under the stochastic volatility model (\ref{B_Epm}) that will  be crucial for the proof of the main results.
 
We start with  a preliminary lemma.
\begin{lemma} 
Consider the model \eqref{B_Epm}. Let $0\leq t \leq s < T$, and $\mathcal G_t:= \mathcal{F}_t \lor \mathcal{F}^{W'}_T$. Then for every $n \geq 0$, there exists $C=C(n,\rho)$ such that  
\begin{equation*}
\begin{split}
\left |\mathbb{E}\left( \frac{\partial^n}{\partial x^n}\left(\frac{\partial^2}{\partial x^2}B_E(s,M_s,k,v_{s}) \right) |  \mathcal G_t \right)\right| & = 
\left |\mathbb{E}\left( \frac{\partial^n}{\partial x^n}\left(\frac{\partial^2}{\partial x \partial k}B_E(s,M_s,k,v_{s}) \right)|  \mathcal G_t \right)\right|  \\
& \leq C \left ( \int_{t}^{T} \phi_r^2 dr \right)^{-\frac{n+1}{2}},
\end{split}
\end{equation*}
where $v_t=\sqrt{\frac{1}{T-t}\int_t^T \phi_r^2 dr}$.
\label{bound}
\end{lemma}

\begin{proof} The proof follows the same steps as the proof of Lemma 6.3.1 in \cite{Alos2021c}. \end{proof}

The main result of this section is the following theorem.
\begin{theorem} \label{decomposition}
Assume Hypotheses \ref{Hyp1}-\ref{Hyp3}. Then, the following relations hold for all $t \in [0,T ]$,
\begin{equation*}
\begin{split}
V^E_t= \mathbb{E}_t\left(B_E(t,S_t,k,v'_t)\right)+\mathbb{E}_t\left(\int_t^T H(s,S_s,k,v'_s)\sigma_s \left(\int_s^T D_s^{W}\sigma_r^2dr\right) ds\right),\\
V^A_t= \mathbb{E}_t\left(B_E(t,M_t,k,v_t)\right)+\mathbb{E}_t\left(\int_t^T H(s,M_s,k,v_s)\phi_s \left(\int_s^T D_s^{W}\phi_r^2dr\right) ds\right),
\end{split}
\end{equation*}
where $H(s,x,k,\sigma)=\frac{1}{2}\partial_{xxx}^3B_E(s,x,k,\sigma)$, $v_t'=\sqrt{\frac{1}{T-t}\int_t^T \sigma_s^2ds}$, and $v_t=\sqrt{\frac{1}{T-t}\int_t^T \phi_s^2ds}$. 
\end{theorem}

\begin{proof} 
The proof follows similar ideas as the proof of Theorem 25 in \cite{Alos2006}. See also Theorem 4.2 in \cite{Alos2007a} and Theorem 6.3.2 in \cite{Alos2021c}.
Notice that, as $B_E(T,x,k,\sigma)=(x-k)_+$ for every $\sigma>0$,  the prices of our European and Asian call options can be written as
\begin{equation*}  
V_t^E =\mathbb{E}_t(B_E(T,S_T,k,v'_T)) \quad \text{and} \quad
V_t^A=\mathbb{E}_t(B_E(T,M_T,k,v_T)),
\end{equation*}
where $v'_T=\lim_{t \rightarrow T} v'_t=\sigma_T$, and $v_T=\lim_{t \rightarrow T} v_t=0$.

We provide the proof of the decomposition formula for $V_t^A$. A similar argument applies for $V_t^E$ and we safely skip it. Applying Theorem \ref{aito} to the function \(B_E(t,M_t,k,v_t)\) and $Y_t=\int_t^T \phi^2_s ds$ noticing that $v_t=\sqrt{\frac{Y_t}{T-t}}$, we obtain
\begin{align*}
&B_E(T,M_T,k,v_T) = B_E(t,M_t,k,v_t)\\
&\qquad +\int_t^T \left(\partial_s B_E(s,M_s,k,v_s)  + \partial_{\sigma}B_E(s,M_s,k,v_s)\frac{v_s^2}{2(T-s)v_s}\right)ds\\
  &\qquad + \int_t^T \partial_x B_E(s,M_s,k,v_s) \phi_sdW_s - \int_t^T \partial_{\sigma}B_E(s,M_s,k,v_s)\frac{\phi_s^2}{2(T-s)v_s}ds  \\
  &\qquad + \int_t^T \partial^2_{\sigma x}B_E(s,M_s,k,v_s)\frac{\phi_s}{2(T-s)v_s}\left(\int_s^T D_s^{W}\phi_r^2dr\right)ds\\
  &\qquad +\frac{1}{2}\int_t^T\partial_{xx}^2 B_E(s,M_s,k,v_s) \phi_s^2 ds.
\end{align*}
Notice that the following relation holds
\begin{align}
\partial_{xx}^2B_E(s,M_s,k,\sigma)=\frac{\partial_{\sigma}B_E(s,M_s,k,\sigma))}{\sigma(T-s)}. 
\label{dvgrel}
\end{align}
Then we get the following
\begin{align*}
&B_E(T,M_T,k,v_T) = B_E(t,M_t,k,v_t)\\
&\qquad+\int_t^T \left(\partial_s B_E(s,M_s,k,v_s)+\frac{1}{2}v_s^2 \partial_{xx}^2B_E(s,M_s,k,v_s) \right)ds\\
  &\qquad+ \int_t^T \partial_x B_E(s,M_s,k,v_s)\phi_sdW_s- \frac{1}{2}\int_t^T \partial_{xx}^2B_E(s,M_s,k,v_s)\phi_s^2ds  \\
  &\qquad+ \frac{1}{2}\int_t^T \partial^3_{xxx}B_E(s,M_s,k,v_s)\phi_s
  \left(\int_s^TD_s^{W}\phi_r^2dr\right) ds \\
  &\qquad+\frac{1}{2}\int_t^T \partial_{xx}^2B_E(s,M_s,k,v_s)\phi_s^2ds.
\end{align*}
The first integral in the above expression is equal to zero due to
equation \eqref{B_Epde}. Finally, taking conditional expectations we conclude that \begin{align*}
&\mathbb{E}_t\left(B_E(T,M_T,k,v_T)\right) = \mathbb{E}_t\left(B_E(t,M_t,k,v_t)\right)\\
  &\qquad+ \mathbb{E}_t\left(\frac{1}{2}\int_t^T \partial_{xxx}^3B_E(s,M_s,k,v_s)\phi_s\left(\int_s^T D_s^{W}\phi_r^2dr \right)ds\right).
\end{align*} 
Observe that by Lemma \ref{bound} and Hypotheses \ref{Hyp1} and \ref{Hyp3}, the last conditional expectation is finite. We also observe that since the function $B_E$ and its derivatives are not bounded, exactly the same truncation argument of Theorem 4.2 in \cite{Alos2007a} can be used here in order to apply Theorem \ref{aito}.
This completes the desired proof.\end{proof}

Based on the result of Theorem \ref{decomposition}, we derive an expression for the ATMIV skew of European and Asian call options under the stochastic volatility model (\ref{B_Epm}).
\begin{proposition}  
Assume Hypotheses \ref{Hyp1}-\ref{Hyp3}. Then, for every $t \in [0,T ]$ the following holds  
\begin{align*}
\begin{split}
\partial_kI_E(t,k_t^*) = \frac{\mathbb{E}_t\left(\int_t^T \partial_k H(s,S_s,k_t^*,v_s')\Lambda_s' ds\right)}{\partial_{\sigma}B_E(t,S_t,k_t^*,I_E(t,k_t^*))},\\
\partial_kI_A(t,k^*_t) = \frac{\mathbb{E}_t\left(\int_t^T \partial_k H(s,M_s,k_t^*,v_s)\Lambda_s ds\right)}{\partial_{\sigma}B_A(t,S_t,y_t,k_t^*,I_A(t,k_t^*))},
\end{split}
\end{align*}
where $\Lambda_s'=\sigma_s\int_s^T D_s^{W}\sigma_r^2dr$ and $\Lambda_s=\phi_s\int_s^T D_s^{W}\phi_r^2dr$.
\label{skew}
\end{proposition}

\begin{proof} We provide the proof for the second part of the theorem. The first part follows by similar arguments. 
Since \(V_t^A=B_A(t,S_t,y_t,k,I_A(t,k))\), differentiating we obtain that
\begin{equation*}\label{st1}
\partial_kV^A_t=\partial_kB_A(t,S_t,y_t,k,I_A(t,k))+\partial_{\sigma}B_A(t,S_t,y_t,k,I_A(t,k))\partial_kI_A(t,k).
\end{equation*}
On the other hand, using Theorem \ref{decomposition}, we 
get that \begin{equation*}\label{st2}
\partial_kV^A_t=\partial_k\mathbb{E}_t\left(B_E(t,M_t,k,v_t)\right)+\mathbb{E}_t\left(\int_t^T \partial_k H(s,M_s,k,v_s)\Lambda_s ds\right).
\end{equation*} 
Combining both  equations, we obtain
that the volatility skew \(\partial_kI_A(t,k)\) is equal to \begin{equation*}
\frac{\mathbb{E}_t\left(\int_t^T \partial_k H(s,M_s,k,v_s)\Lambda_s ds\right)+\mathbb{E}_t\left(\partial_kB_E(t,M_t,k,v_t)\right)-\partial_kB_A(t,S_t,y_t,k,I_A(t,k))}{\partial_{\sigma}B_A(t,S_t,y_t,k,I_A(t,k))}.
\label{st3}
\end{equation*}
Finally, using the fact that
$$
\partial_kB_E(t,M_t,k_t^*,\sigma)=\partial_kB_A(t,S_t,y_t,k_t^*,\sigma)=-\frac{1}{2}
$$
we complete the desired proof.
\end{proof}

In order to compute the limit of ATMIV skew, we need to identify the  leading order terms in the numerator of the formulas obtained in Proposition \ref{skew}, for which the next result will be crucial.
\begin{proposition}\label{skew2}
Assume Hypotheses \ref{Hyp1}-\ref{Hyp3}. Then, for all $t \leq T$,
\begin{align*}
\begin{split}
\mathbb{E}_t\left(\int_t^T G(s,S_s,k,v'_s)\Lambda_s' ds\right)&=\mathbb{E}_t\left( G(t,S_t,k,v'_t) J_t'\right)\\
&+\mathbb{E}_t\left(\frac{1}{2}\int_t^T \partial_{xxx}^3G(s,S_s,k,v'_s) J_s'\Lambda_s'ds\right)\\
&+ \mathbb{E}_t\left(\int_t^T \partial_{x}G(s,S_s,k,v'_s)\sigma_s D^{-}J_s'ds\right),\\
\mathbb{E}_t\left(\int_t^T G(s,M_s,k,v_s)\Lambda_s ds\right)&=\mathbb{E}_t\left( G(t,M_t,k,v_t) J_t\right)\\
&+\mathbb{E}_t\left(\frac{1}{2}\int_t^T \partial_{xxx}^3G(s,M_s,k,v_s) J_s\Lambda_sds\right)\\
&+ \mathbb{E}_t\left(\int_t^T \partial_{x}G(s,M_s,k,v_s)\phi_s D^{-}J_sds\right),
\end{split}
\end{align*}
where $G(t,x,k,\sigma)=\partial_k H(t,x,k,\sigma)$, $J_t=\int_t^T \Lambda_s ds$, $J'_t=\int_t^T \Lambda'_s ds$, $D^{-}J'_s =\int_s^T D_s^W \Lambda'_r dr$ and $D^{-}J_s =\int_s^T D_s^W \Lambda_r dr$.
\label{LeadOrder}
\end{proposition}

\begin{proof} We only prove the second part of the theorem since the first part follows by similar arguments. Applying Theorem \ref{aito} to the function $\partial_k H(t,M_t,k,v_t) \int_t^T \Lambda_s ds$,
we obtain that \begin{align*}
&\int_t^TG(s,M_s,k,v_s)\Lambda_sds = G(t,M_t,k,v_t)J_t \\
  &\qquad +\int_t^T \left(\partial_sG(s,M_s,k,v_s)+\frac{v_s^2}{2(T-s)v_s}\partial_vG(s,M_s,k,v_s)\right)J_sds\\
  &\qquad  + \int_t^T \partial_x G(s,M_s,k,v_s)J_s \phi_sdW_s \\
  &\qquad   - \int_t^T \partial_vG(s,M_s,k,v_s)J_s\frac{\phi_s^2}{2(T-s)v_s}ds   + \int_t^T \partial_{vx}^2G(s,M_s,k,v_s)J_s\Lambda_s\frac{1}{2(T-s)v_s}ds\\
  &\qquad  + \int_t^T \partial_{x}G(s,M_s,k,v_s)\phi_s D^{-}J_sds+\frac{1}{2}\int_t^T \phi_s^2\partial_{xx}^2 G(s,M_s,k,v_s)J_sds.
\end{align*}
Next, equations \eqref{B_Epde} and (\ref{dvgrel}) imply that
\begin{align*}
&\partial_s G(s,M_s,k,v_s) + \frac{1}{2}v_s^2 \partial^2_{xx} G(s,M_s,k,v_s)=0,\\
&\partial_{xx}^2G(s,M_s,k,v_s)=\frac{\partial_{v}G(s,M_s,k,v_s)}{v_s(T-s)}.
\end{align*}
Using the above equations we get that \begin{align*}
\begin{split}
&\int_t^TG(s,M_s,k,v_s)\Lambda_sds = G(t,M_t,k,v_t)J_t +\int_t^T \partial_x G(s,M_s,k,v_s)J_s\phi_sdW_s\\
  &\qquad + \frac{1}{2}\int_t^T \partial_{xxx}^3G(s,M_s,k,v_s)J_s\Lambda_sds+ \int_t^T \partial_{x}G(s,M_s,k,v_s)\phi_s D^{-}J_sds.
  \end{split}
\end{align*}
Finally, taking conditional expectations and noticing that by Lemma \ref{bound} and Hypotheses \ref{Hyp1} and \ref{Hyp4}, all conditional expectations are finite, we complete the desired proof. Remark that as for Theorem \ref{decomposition}, the same truncation argument as in \cite{Alos2007a} can be used in order to apply Theorem \ref{aito}.
\end{proof}

\section{Proof of Theorem \ref{limskew}} 

\subsection{Proof of (\ref{main1}) in Theorem \ref{limskew}: ATM implied volatility
level}
This section is devoted to the proof of (\ref{main1}) in Theorem \ref{limskew}. We only show the result for  $I_A( 0,k^*)$ since the proof for $I_E(0,k^*)$ follows along the same lines.
We start proving the result for the implied  volatility in the uncorrelated case ($\rho=0$) that we  denote by $I_A^0(t,k)$.

\subsubsection{The uncorrelated case}

We aim to apply the decomposition for the option price obtained in Theorem \ref{decomposition}.
Observe that
$$
D_s^{W}\phi_r^2dr= \frac{(T-r)^2}{T^2} D_s^ W \sigma_r^2=\frac{(T-r)^2}{T^2} 2 \sigma_r D_s^W \sigma_r
=\frac{(T-r)^2}{T^2} 2 \sigma_r \rho D_s^{W'} \sigma_r.
$$
Thus, if $\rho=0$, the decomposition formula give us
that $V_{0}^A= \E(B_E( 0,M_0,k,v_{0}))$.
Then the ATMIV satisfies that 
\begin{equation*}
\begin{split} \label{dos}
&I_A^0( 0,k^*) = (B_A)^{-1}(0, S_0, y_0, k^*,V_{0}^A)  = \mathbb{E}\left( B_A^{-1}(0, S_0, y_0,k^*,\mathbb{E} B_E( 0,M_0,k^*,v_{0}) ) \right)  \\
& \quad= \mathbb{E}\left( B_A^{-1}(0, S_0, y_0,k^*,\mathbb{E} B_E( 0,M_0,k^*,v_{0})) \pm B_A^{-1}(0, S_0, y_0, k^*,B_E( 0,M_0,k^*,v_{0})) \right)   \\
& \quad=\mathbb{E} \left(B_A^{-1}(0, S_0, y_0,k^*,\Phi_0 ) -B_A^{-1}(0, S_0, y_0,k^*,\Phi_T )\right) +\sqrt{3}\mathbb{E}(v_{0})  \\
&\quad =\sqrt{3}\mathbb{E} \left( B_E^{-1}(0, M_0,k^*,\Phi_0) -B_E^{-1}(0, M_0,k^*,\Phi_T ) \right) +\sqrt{3}\mathbb{E}(v_{0}),
\end{split}
\end{equation*}
where $\Phi_{r}:=\E_{r}\left(B_E\left( 0,M_0,k^*,v_{0}\right) \right)$. The last two lines follow from the observation that $B_A^{-1}(t, S_t, y_t, k^*_t,B_E( t,M_t,k^*_t,\sigma )) =\frac{T\sqrt{3}}{T-t}\sigma$. 

Observe that  as $\rho=0$, the Brownian motions $W$ and $W'$ are independent. Thus, $\Phi_r = \E\left( B_E\left( 0,M_0,k^*,v_{0}\right)| \mathcal{F}_r^{W'} \right)$ and $(\Phi_r)_{r \geq 0}$ is a martingale wrt to the filtration $(\mathcal{F}^{W'}_r)_{r \geq 0}$. By the martingale representation  theorem, there exists a square integrable and  $\mathcal{F}^{W'}$-adapted  process $(U_r)_{r\geq 0}$ such  that
$$\Phi_r=\Phi_0+\int_0^r U_s dW_s'.$$ 
A direct application of the classical It\^o's formula gives 
\begin{equation*}
\begin{split}
&\E\left( B_E^{-1}(0, M_0,k^*,\Phi_0 )-B_E^{-1}(0, M_0, k^*,\Phi_T )\right) \\
&=-\E\left(\int_{0}^{T}( B_E^{-1}) ^{\prime }(0, M_0, k^*, \Phi_s ) U_{s}dW_{s}'
+\frac{1}{2}\int_{0}^{T}( B_E^{-1}) ^{\prime \prime}(0, M_0, k^*, \Phi_s) U_{s}^{2}ds\right), 
\end{split}
\end{equation*}
where $(B_E^{-1})^{\prime}$ and $(B_E^{-1})^{\prime \prime}$ denote, respectively, the first and  second  derivative of $B_E^{-1}$ with respect to $\sigma$. A direct computation gives $$B_E^{-1}(0, M_0, k^*, \Phi_s)=\frac{ \Phi_s\sqrt{2 \pi}}{\sqrt{T}} \quad \text{and} \quad ((B_E)^{^-1})^{\prime \prime }(0, M_0, k^*, \Phi_s)=0,$$ which leads to 
$$
I_A^0( 0,k^*) =\sqrt{3}\mathbb{E}(v_0).
$$
By continuity, we conclude that
\begin{equation} \label{corre}
\lim_{T \rightarrow 0} I_A^0( 0,k^*) = \sigma_0,
\end{equation}
which proves (\ref{main1}) in the uncorrelated case.

\subsubsection{The correlated case}

Using similar ideas as in the uncorrelated case we get that
\begin{equation*}
\begin{split} 
I_A( 0,k^*) &= (B_A)^{-1}(0, S_0, y_0, k^*,V_{0}^A) \\
& =  \mathbb{E}\left( (B_A)^{-1}(0, S_0, y_0,k^*,\Gamma_T) \pm (B_A)^{-1}(0, S_0, y_0,k^*,\Gamma_0) \right)   \\
&=\mathbb{E} \left((B_A)^{-1}(0, S_0, y_0,k^*,\Gamma_T ) -(B_A)^{-1}(0, S_0, y_0,k^*,\Gamma_0)\right) +I_A^0( 0,k^*)  \\
&=\sqrt{3}\mathbb{E} \left( B_E^{-1}(0, M_0,k^*,\Gamma_T) -B_E^{-1}(0, M_0,k^*,\Gamma_0 ) \right) + I_A^0( 0,k^*),
\end{split}
\end{equation*}
where $\Gamma_{s}:=\mathbb{E}[ B_E(0,M_0,k^*,v_{0})] +\frac{\rho }{2}\mathbb{E}\left(\int_{0}^{s} H(r,M_{r},k^*_t,v_{r})\Lambda _{r}dr\right)$.

A direct application of It\^o's formula gives 
\begin{equation*}
\begin{split}
I_A( 0,k^*) &= I_A^0( 0,k^*)  + \mathbb{E} \left(\int_{0}^{T}( B_E^{-1}) ^{\prime }(0, M_0, k^{\ast},\Gamma _{s}) H(s,M_s,k^*_t,v_{s})\Lambda _{s}ds\right)  \\
&=I_A^0( 0,k^*)  + \frac{\sqrt{2 \pi}}{\sqrt{T}}\mathbb{E} \left(\int_{0}^{T} H(s,M_s,k^*_t,v_{s})\Lambda _{s}ds\right).
\end{split}
\end{equation*}

Next, Hypotheses \ref{Hyp1} and \ref{Hyp4} together with Lemma \ref{bound} imply that
\begin{equation*}
\begin{split}
&\bigg\vert  \mathbb{E} \left(\int_{0}^{T} H(s,M_s,k^*,v_{s})\Lambda _{s}ds\right) \bigg\vert \leq C\mathbb{E} \left(\left(\int_0^T\phi_r^2dr\right)^{-1}\left(\int_{0}^{T} \vert \Lambda_{s} \vert  ds\right)\right) \\
&\qquad\leq C  \left(\int_0^T \frac{(T-r)^2}{T^2}dr\right)^{-1}\int_{0}^{T} \frac{(T-s)^3}{T^3}\int_s^T \mathbb{E}(\vert D_{s}^{W'}\sigma _{r} \vert)  drds\\
&\qquad\leq C \frac{1}{T^4}\int_0^T (T-s)^3\int_s^T(r-s)^{H-\frac{1}{2}}drds \\
&\qquad = C T^{H+\frac{1}{2}}.
\end{split}
\end{equation*}
Thus, we conclude that $\frac{1}{\sqrt{T}}\mathbb{E} \left(\int_{0}^{T} H(s,M_s,k^*_t,v_{s})\Lambda _{s}ds\right) \to 0$ as $T\to 0$.
Finally, using (\ref{corre}), we get that
$$
I_A( 0,k^*)\to \sigma_0
$$
as $T\to 0$, which concludes the proof of (\ref{main1}).

\subsection{Proof of (\ref{main2}) in Theorem \ref{limskew}: ATMIV
skew}\label{limiting-behaviour-of-the-at-the-money-implied-volatility-skew}

We provide the proof for the Asian call option case. The result for the European call follows by the same argument. 

Appealing to Propositions \ref{skew} and \ref{skew2} we have that
\begin{equation}\label{aux1} \begin{split}
\partial_kI_A(0,k^*) &= \frac{1}{\partial_{\sigma}B_A(0,S_0,y_0,k^*,I_A(0,k^*))}
\bigg\{\mathbb{E}\left( G(0,M_0,k^*,v_0) J_0\right)\\
&\qquad +\mathbb{E}\left(\frac{1}{2}\int_0^T \partial_{xxx}^3G(s,M_s,k^*,v_s) J_s\Lambda_sds\right)\\
&\qquad + \mathbb{E}\left(\int_0^T \partial_{x}G(s,M_s,k^*,v_s)\phi_s D^{-}J_sds\right)\bigg\}.
\end{split}
\end{equation}
We start bounding the second term in (\ref{aux1}). Using Hypotheses \ref{Hyp1} and \ref{Hyp4} together with Cauchy-Schwarz inequality we get that
\begin{equation*}
\begin{split}
&\mathbb{E}(\vert J_s \Lambda_s \vert)  \\
& \leq C \frac{(T-s)}{T^6} \mathbb{E}\left(\left(\int_s^T (T-r)^2\vert D_s^{W'}\sigma_r \vert   dr\right) \left(\int_s^T(T-u)^3\int_u^T \vert D_u^{W'}\sigma_r \vert dr du\right)\right) \\
&\leq C \frac{(T-s)}{T^6} \sqrt{\mathbb{E}\left(\left(\int_s^T (T-r)^2\vert D_s^{W'}\sigma_r \vert dr\right)^2\right)\mathbb{E}\left(\left(\int_s^T (T-u)^3\int_u^T \vert  D_u^{W'}\sigma_r\vert  dr du\right)^2\right)}\\
& \leq C \frac{(T-s)}{T^6} \sqrt{ (T-s)^5\int_s^T\mathbb{E}\left(\vert D_s^{W'}\sigma_r \vert^2\right) dr
(T-s)^8\int_s^T\int_u^T \mathbb{E}\left(\vert  D_u^{W'}\sigma_r \vert^2  \right) dr du}\\
&\leq  C \frac{(T-s)^{2H+8}}{T^6}.
\end{split}
\end{equation*}
Then, using Lemma \ref{bound} we conclude that 
\begin{equation*}
\begin{split}
&\bigg\vert \mathbb{E}\left(\frac{1}{2}\int_0^T \partial_{xxx}^3G(s,M_s,k,v_s) J_s\Lambda_sds \right) \bigg\vert \\
&\qquad \leq C \frac{1}{T^3} \int_0^T \frac{(T-s)^{2H+8}}{T^6} dt = C T^{2H}.
\end{split}
\end{equation*}

We next bound the third term in (\ref{aux1}). We  have that
\begin{align*}
\begin{split}
\int_s^TD_s^W\Lambda_rdr &= \int_s^TD_s^W\left(\phi_r\int_r^TD_r^W\phi_u^2du\right)dr\\
&= \int_s^T\left(\left(D_s^W\phi_r\right)\int_r^TD_r^W\phi_u^2du+\phi_r\int_r^TD_s^WD_r^W\phi_u^2du\right)dr,
\end{split}
\end{align*}
where
$$
 D_s^WD_r^W\phi_u^2   =2(D_s^W\phi_uD_r^W\phi_u+\phi_uD_s^WD_r^W\phi_u).
$$
Hypothesis \ref{Hyp1} implies that
$$\vert D_s^WD_r^W\phi_u^2 \vert \leq C \frac{(T-u)^2}{T^2}\left( \vert D_s^{W'}\sigma_u \vert \, \vert D_r^{W'}\sigma_u \vert + \vert D_s^{W'}D_r^{W'}\sigma_u \vert  \right).$$
Then, appealing to Hypothesis \ref{Hyp4}, we get that 
\begin{equation*}
\begin{split}
\mathbb{E}\left(\bigg\vert\phi_r\int_r^TD_s^WD_r^W\phi_u^2du \bigg\vert\right) & \leq C \frac{(T-r)^3}{T^3} \int_r^T  (u-r)^{H-\frac{1}{2}}(u-s)^{H-\frac{1}{2}} du \\
&\leq C \frac{(T-r)^{3}}{T^3} (T-s)^{2H}
\end{split}
\end{equation*}
and
\begin{equation*}
\begin{split}
\mathbb{E}\left(\bigg\vert D_s^W\phi_r\int_r^TD_r^W\phi_u^2du \bigg\vert\right) &\leq C \frac{(T-r)^3}{T^3} 
\mathbb{E}\left( \vert  D_s^{W'}\sigma_r \vert \int_r^T \vert D_r^{W'}\sigma_u \vert du\right) \\
&\leq C \frac{(T-r)^3}{T^3}(r-s)^{H-\frac{1}{2}}(T-r)^{H+\frac{1}{2}}.
\end{split}
\end{equation*}
Therefore, we conclude that 
\begin{equation*}
\begin{split}
\mathbb{E}\left(\bigg\vert\int_s^TD_s^W\Lambda_rdr \bigg\vert\right) & \leq C\int_s^T \frac{(T-r)^3}{T^3}\left( (r-s)^{H-\frac{1}{2}}(T-s)^{H+\frac{1}{2}} + (T-s)^{2H}\right) dr \\
& \leq C \frac{(T-s)^{2H+4}}{T^3}.
\end{split}
\end{equation*}
Finally, using Lemma \ref{bound} we obtain that the  third term in (\ref{aux1}) satisfies that
\begin{equation*}
\begin{split}
\bigg\vert\mathbb{E}\left(\int_0^T \partial_{x}G(s,M_s,k,v_s)\phi_s D^{-}J_sds\right) \bigg\vert &\leq C  \frac{1}{T^2}\int_0^T  \frac{(T-s)^{5+2H}}{T^4} ds\\
&= CT^{2H}.
\end{split}
\end{equation*}
Taking into account the relation $\partial_{\sigma}B_A(0,S_0,y_0, k^*,\sigma) = \frac{1}{\sqrt{3}}\partial_{\sigma}B_E(0,S_0,k^*,\sigma)$ we get that
$$
\partial_{\sigma}B_A(0,S_0,y_0,k^*,\sigma)=\frac{\sqrt{T} }{\sqrt{6 \pi }}.
$$
Thus, as
$$
\lim_{T \rightarrow 0} T^{\max(\frac12-H,0)} \frac{1}{\sqrt{T}} T^{2H}=0,
$$
the second and third terms in (\ref{aux1}) will contribute as $0$ in the limit (\ref{main22}).

We  finally study the  first term in (\ref{aux1}). 
Direct differentiation give us  that
$$
G(0,M_0,k^*,v_0)= \frac{1}{2 \sqrt{2 \pi} v_0^3 T^{3/2}}.
$$
Thus, by  (\ref{aux1}), we get that\begin{equation*}
\begin{split}
&\lim_{T \to 0} T^{\max(\frac{1}{2}-H, 0)}\partial_kI_A(0,k^*) \\
&= \sqrt{3} \rho\lim_{T \to 0}  T^{\max(\frac{1}{2}-H, 0)} \frac{1}{T^5}
 \mathbb{E}\left( v_0^{-3} \int_0^T  \sigma_r (T-r)\int_r^T \sigma_u (T-u)^2 D_r^{W'}\sigma_u  du dr\right).
\end{split}
\end{equation*}

Next, writing $v_0^{-3}=(v_0^{-3}-\sigma_0^{-3} \sqrt{3}^3)+\sigma_0^{-3} \sqrt{3}^3$, the above limit can be written as  $A_1+A_2$, where
\begin{equation*}
\begin{split}
A_1 &= \sqrt{3} \rho\lim_{T \to 0}  T^{\max(\frac{1}{2}-H, 0)} \frac{1}{T^5}
 \mathbb{E}\left( (v_0^{-3}-\sigma_0^{-3} \sqrt{3}^3)F_T \right), \\
 A_2 &= \sqrt{3} \rho\lim_{T \to 0}  T^{\max(\frac{1}{2}-H, 0)} \frac{1}{T^5}
 \mathbb{E}\left( \sigma_0^{-3} \sqrt{3}^3 F_T\right),
\end{split}
\end{equation*}
and
$$
F_T=\int_0^T  \sigma_r (T-r)\int_r^T \sigma_u (T-u)^2 D_r^{W'}\sigma_u  du dr.
$$

We start showing that $A_1=0$. By Cauchy-Schwarz inequality,
\begin{equation*} \begin{split}
&T^{\max(\frac{1}{2}-H, 0)} \frac{1}{T^5}
 \mathbb{E}\left( (v_0^{-3}-\sigma_0^{-3} \sqrt{3}^3)F_T \right)\\
 &\leq T^{\max(\frac{1}{2}-H, 0)} \frac{1}{T^5}
 (\mathbb{E} \vert v_0^{-3}-\sigma_0^{-3} \sqrt{3}^3\vert^2)^{1/2} (\E\vert F_T \vert^2 )^{1/2}.
 \end{split}
\end{equation*}
By  continuity, it holds that $v_0^{-3}$ converges  towards $\sigma_0^{-3} \sqrt{3}^3$ a.s. as $T \rightarrow 0$. Thus, by dominated convergence the limit as $T \rightarrow 0$ of $(\mathbb{E} \vert v_0^{-3}-\sigma_0^{-3} \sqrt{3}^3\vert^2)^{1/2}$ is $0$. 
On  the other hand, by Hypotheses \ref{Hyp1} and (\ref{d1}), we have that
$$
T^{\max(\frac{1}{2}-H, 0)} \frac{1}{T^5}(\E\vert F_T \vert^2 )^{1/2} \leq C T^{\max(\frac{1}{2}-H, 0)} T^{H-\frac12},
$$
which is bounded, thus we have shown that $A_1=0$. We are left to show that
$$
A_2=\lim_{T \to 0} T^{\max(\frac12-H, 0)}\frac{9 \rho}{ \sigma_0 T^5}\int_0^T (T-r)\int_r^T(T-u)^2\E(D_r^{W'}\sigma_u) du  dr,
$$
which is the limit in (\ref{main22}), which is finite by assumption. This follows by writing $\sigma_r=(\sigma_r-\sigma_0)+\sigma_0$ and $\sigma_u=(\sigma_u-\sigma_0)+\sigma_0$
and proceeding as above using the a.s. continuity of $\sigma_t$ and dominated convergence to show that the first term in each decomposition converges to zero. 
This concludes the proof of (\ref{main22}).

\section{Numerical analysis}\label{numerical-analysis}

In this section we apply the  theoretical results presented earlier to some examples of stochastic volatility models. We justify our findings with numerical simulations.

\subsection{The SABR model}\label{the-sabr-model}

We consider the SABR stochastic volatility model with skewness parameter 1, which is the most common case
from a practical point of view. This corresponds to equation (\ref{B_Epm}), where $S_t$ denotes the  forward price of the underlying asset and
$$
    d\sigma_t  = \alpha\sigma_t dW_t', \qquad \sigma_t=\sigma_0e^{\alpha W'_t-\frac{\alpha^2}{2}t}.
  $$ 
  where $\alpha>0$ is the volatility of volatility.

Notice that this model does not satisfy Hypothesis 1, so a truncation argument  is needed in order to check that Theorem \ref{limskew} is true for this model. We skip it here for the sake of simplicity since it is identical to the one presented in \cite{mp2022}.
 
For \(r\leq t\), we have that
$
D_r^{W'}\sigma_t = \alpha \sigma_t
$
and $\mathbb E \left(D_r^{W'}\sigma_t\right) = \alpha \sigma_0$. Therefore, applying Theorem \ref{limskew} we conclude that
\begin{align*}
\begin{split}
\lim_{T \to 0} \partial_k I_A(0,k^{*}) &= \frac{3}{5}\rho \alpha,\\
\lim_{T \to 0} \partial_k I_E(0,k^{*}) &= \frac{1}{2}\rho \alpha.
 \end{split}
 \end{align*}

We next proceed with some numerical simulations using the following parameters
$$S_0 = 10, \, T=\frac{1}{252}, \, dt=\frac{T}{50}, \, \alpha = 0.5, \, \rho = -0.3, \, \sigma_0 =(0.1, 0.2, \dots, 1.4).$$

In order to get estimates of an Asian call option we use antithetic variates. The estimator of the price is defined as follows
\begin{align}
\begin{split}
\hat{V}_{sabr} &= \frac{ \frac{1}{N}\sum_{i=1}^N V_T^i + \frac{1}{N}\sum_{i=1}^N V_T^{i,A} }{2},
\end{split}
\label{antithetic}
\end{align}
where $N=2000000$ and the sub-index $A$ denotes the value of an Asian call option computed on the antithetic trajectory of a Monte Carlo path.

In order to retrieve the implied volatility we use method presented in \cite{jaeckel}. For the estimation of the skew, we use the following expression which allows us to avoid a finite difference method
\begin{align}
\begin{split}
\partial_k \hat{I}^A(0,k^{*}) &= \frac{-\partial_kB_A(0,X_0,y_0,k^*,I_A(0,k^*))-\mathbb{E}\left(1_{A_T\geq k^*}\right)}{\partial_{\sigma}B_A(0,X_0,y_0,k^*,I_A(0,k^*))}=\frac{\frac{1}{2}-\mathbb{E}\left(1_{A_T\geq k^*}\right)}{\sqrt{\frac{T}{6 \pi}}},\\
\partial_k \hat{I}^E(0,k^{*}) &= \frac{-\partial_k B_E(0,S_0,k^*,I_E(0,k^*))-\mathbb{E}\left(1_{S_T\geq k^*}\right)}{\partial_{\sigma}B_E(0,S_0,k^*,I_E(0,k^*))}=\frac{\frac{1}{2}-\mathbb{E}\left(1_{S_T\geq k^*}\right)}{\sqrt{\frac{T}{2 \pi}}}.
\end{split}
\label{skew_estimator2}
\end{align}

We use equation \eqref{skew_estimator2} in order to get estimates of the skew. In Figure \ref{fig2} we present the results of a Monte Carlo simulation
which aims to evaluate numerically the level and the skew of the at-the-money implied volatility of an Asian call option under the SABR model. And in Figure \ref{fig22} we do the same, but for the European call option. We observe  that all the numerical results fit the theoretical ones.
\begin{figure}[h]
\centering
\begin{tabular}{ccc}
  \includegraphics[width=60mm]{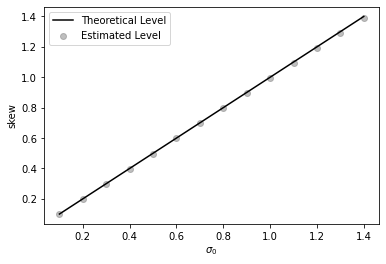} &   \includegraphics[width=60mm]{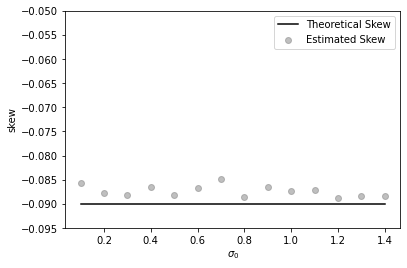} \\[2pt]
\end{tabular}
\caption{At-the-money level and the skew of the Implied Volatility of an Asian call under the SABR model.}
\label{fig2}
\end{figure}

\begin{figure}[h]
\centering
\begin{tabular}{ccc}
  \includegraphics[width=60mm]{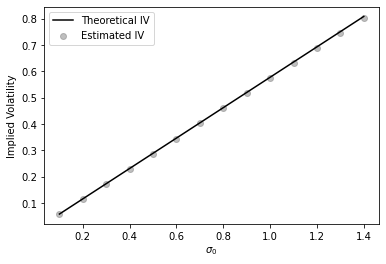}  &
   \includegraphics[width=60mm]{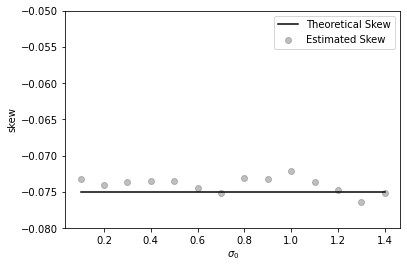} \\[2pt]
   \end{tabular}
\caption{At-the-money level and the skew of the Implied Volatility of the European call option under the SABR model.}
\label{fig22}
\end{figure}

\subsection{The fractional Bergomi model}\label{fractional-bergomi-model}

The fractional Bergomi stochastic volatility model asssumes equation (\ref{B_Epm}) with
\begin{equation*} \begin{split}
    \sigma_t^2 &= \sigma_0^2 e^{v \sqrt{2H}Z_t-\frac{1}{2}v^2t^{2H}}, \\
    Z_t& = \int_0^t(t-s)^{H-\frac{1}{2}}dW'_s,
    \end{split}
    \end{equation*}
    where $H \in  (0,1)$ and $v>0$.
    
As for the SABR model, a truncation argument as in \cite{mp2022} is needed in order to apply the results in the previous sections, as Hypothesis 1 is not satisfied.
Moreover, for $r \leq u$, we have 
\begin{equation*} 
\begin{split}
D_r^{W'}\sigma_u &=  \frac12 \sigma_u v \sqrt{2 H} (u-r)^{H-\frac{1}{2}}, \\
\E(D_r^{W'}\sigma_u) &=e^{-\frac18 v^2 u^{2H}}\frac12  \sigma_0
v \sqrt{2 H} (u-r)^{H-\frac{1}{2}},
\end{split}
\end{equation*}
which gives
\begin{equation} \label{rou}
\begin{split}
\lim_{T \to 0} \partial_kI_A(0,k^{*}) 
 &=\begin{cases}
0 \quad &\text{if} \quad H>\frac12 \\
\frac{3\rho v }{10}  \quad &\text{if} \quad H=\frac12,
\end{cases},\\
\lim_{T \to 0} \partial_kI_E(0,k^{*}) 
 &=\begin{cases}
0 \quad &\text{if} \quad H>\frac12 \\
\frac{\rho v }{4}\quad &\text{if} \quad H=\frac12.
\end{cases}
\end{split}
\end{equation}
and  for $H<\frac12$
\begin{equation} \begin{split}\label{rou2}
\lim_{T \to 0} T^{\frac{1}{2}-H}\partial_k I_A(0,k^{*})&=\frac{288\rho v \sqrt{2 H}  }{(2 H+9) \left(8 H^3+36 H^2+46 H+15\right)},\\
\lim_{T \to 0} T^{\frac{1}{2}-H}\partial_k I_E(0,k^{*})&=\frac{2\rho v \sqrt{2 H} }{ \left(3+ 4H(2+H)\right)}.
\end{split}
\end{equation}

The parameters used for the Monte Carlo simulation are the following
$$
S_0 = 100, \, T=0.001, \, dt=\frac{T}{50}, \, H=(0.4, 0.7), \, v = 0.5, \, \rho = -0.3, \, \sigma_0 =(0.1, 0.2, \dots, 1.4).$$

In order to get estimates of the price of an Asian call option under the fractional Bergomi model we use antithetic variates presented in equation \eqref{antithetic}. Then, the level of at-the-money implied volatility of an Asian and European call options are presented on Figures \ref{fig5} and \ref{fig55}, respectively. One can see that the result is independent of $H$.
\begin{figure}[h]
\centering
\begin{tabular}{ccc}
  \includegraphics[width=60mm]{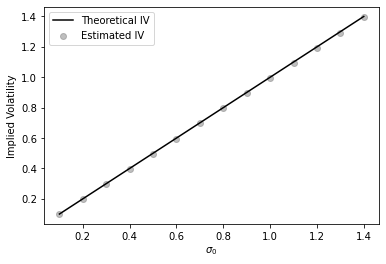} &   \includegraphics[width=60mm]{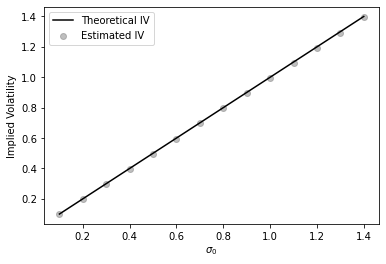} \\
(a) H=0.4 & (b) H=0.7 \\[2pt]
\end{tabular}
\caption{At-the-money level of the IV of an Asian call under fractional Bergomi model.}
\label{fig5}
\end{figure}
\begin{figure}[h]
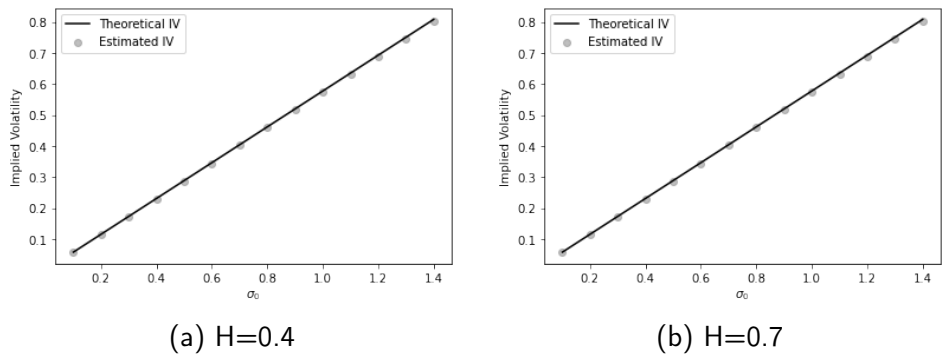

\centering
\begin{tabular}{ccc}
  \includegraphics[width=60mm]{IVlevelERB04} &   \includegraphics[width=60mm]{IVlevelERB07} \\
(a) H=0.4 & (b) H=0.7 \\[2pt]
\end{tabular}
\caption{At-the-money level of the IV of the European call under fractional Bergomi model.}
\label{fig55}
\end{figure}

On Figures \ref{fig6} and \ref{fig66} we present the ATM implied volatility skew as a function of maturity of an Asian and European call options, respectively, for two different values of $H$, where we observe the blow up to $-\infty$ for the case $H=0.4$. Equation \eqref{rou} suggests that the theoretical values of the slope of the at-the-money skew (as a function $\alpha T^{\beta}$) in the cases of Asian and European call options with $H < \frac{1}{2}$ are $-0.0497$ and $-0.0336$, respectively. We conclude that theoretical results are in line with the observed numbers.
\begin{figure}[h]
\centering
\begin{tabular}{ccc}
  \includegraphics[width=60mm]{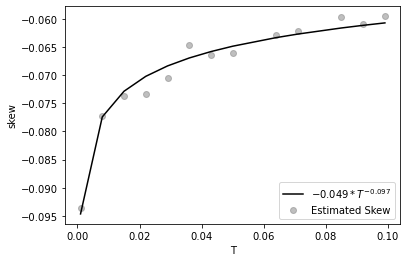} &   \includegraphics[width=60mm]{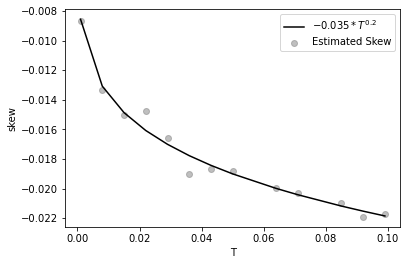} \\
(a) H=0.4, $\sigma_0=0.3$ & (b) H=0.7, $\sigma_0=0.3$  \\[2pt]
\end{tabular}
\caption{At-the-money IV skew of an Asian call as a function of $T$ under fractional Bergomi model.}
\label{fig6}
\end{figure}
\begin{figure}[h]
\centering
\begin{tabular}{ccc}
  \includegraphics[width=60mm]{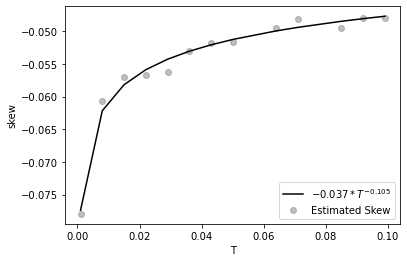} &   \includegraphics[width=60mm]{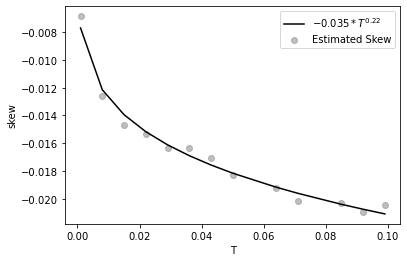} \\
(a) H=0.4, $\sigma_0=0.3$ & (b) H=0.7, $\sigma_0=0.3$  \\[2pt]
\end{tabular}
\caption{At-the-money IV skew of the European call as a function of $T$ under fractional Bergomi model.}
\label{fig66}
\end{figure}
Due to the blow up of the at-the-money implied volatility skew of an Asian and European call options when $H<\frac{1}{2}$ we also plot (as a function of $\sigma_0$) the quantities $T^{\frac{1}{2}-H}\partial_k \hat{I}^A(0,k^{*})$ and $T^{\frac{1}{2}-H}\partial_k \hat{I}^E(0,k^{*})$ for $H=0.4$ as well as $\partial_k \hat{I}^A(0,k^{*})$ and $\partial_k \hat{I}^E(0,k^{*})$  for $H=0.7$ for fixed value of $T=0.001$. The result is presented at Figures \ref{fig7} and \ref{fig77}. We see that numerical estimates agree with the presented theoretical findings.  
\begin{figure}[h]
\centering
\begin{tabular}{ccc}
  \includegraphics[width=60mm]{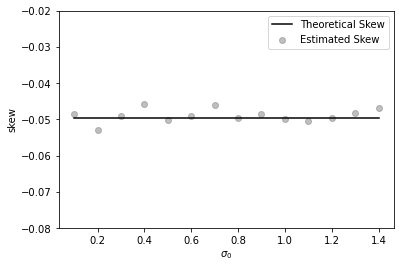} &   \includegraphics[width=60mm]{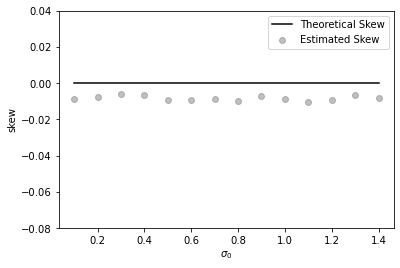} \\
(a) H=0.4 & (b) H=0.7 \\[2pt]
\end{tabular}
\caption{At-the-money IV skew of an Asian call as a function of $\sigma_0$ under fractional Bergomi model}
\label{fig7}
\end{figure}
\begin{figure}[h]
\centering
\begin{tabular}{ccc}
  \includegraphics[width=60mm]{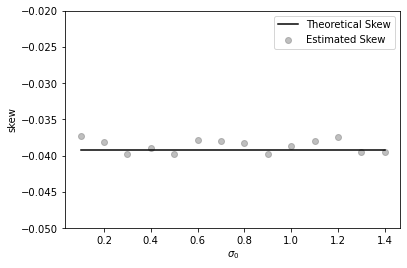} &   \includegraphics[width=60mm]{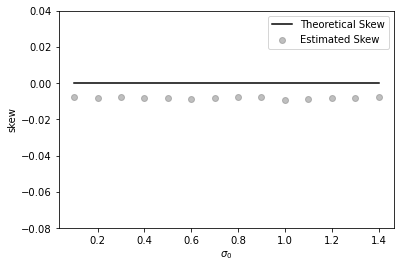} \\
(a) H=0.4 & (b) H=0.7 \\[2pt]
\end{tabular}
\caption{At-the-money IV skew of the European call as a function of $\sigma_0$ under fractional Bergomi model}
\label{fig77}
\end{figure}

\subsubsection{Local volatility model}\label{local-volatility-model}

We  consider the local volatility model
\begin{equation*} 
dS_t=\sigma(S_t) dW_t,
\end{equation*}
where $\sigma \in C^2_b$ (bounded with bounded first and second derivatives) and $\sigma (x)  \geq c >0$, for all $x$. 
We apply Theorem \ref{limskew} when $S_t$ follows this model with $\rho=0$ and $\sigma_t=\sigma(S_t)$. Under the above assumptions it is easy to see that all the hypotheses are satisfied with $H=\frac12$. Thus, for the ATMIV  level, we directly see that the limit is equal to $\sigma(S_0)$. For the ATMIV skew, we need to compute $D_r \sigma(S_t)$.  We have for $r \leq u$,
\begin{equation*} \begin{split}
D^W_r \sigma(S_u)=\sigma'(S_u) D^W_rS_u=\sigma'(S_u) \left(\sigma(S_r) +\int_r^u D_r^W\sigma(S_s) dW_s \right).
\end{split}
\end{equation*}
In particular,
\begin{equation*}
\E\left(D_r \sigma(S_u) \right)=\E\left(\sigma'(S_u)\sigma(S_r)\right).
\end{equation*}
This can be written as 
\begin{equation*} \begin{split}
\E\left(D_r \sigma(S_u) \right)=\sigma'(S_0) \sigma(S_0) &+\E\left((\sigma'(S_u)-\sigma'(S_0)) \sigma(S_r)\right)\\
& +\sigma'(S_0) \E\left((\sigma(S_r)-\sigma(S_0))\right).
\end{split}
\end{equation*}
Then,  using the mean value theorem and the fact that $S_t$ has H\"older continuous sample paths of any order $\gamma<\frac12$, we see that the last two terms of the last display will not contribute in the limit (\ref{main2}) and (\ref{main22}). Thus, we get
 \begin{align*}
  \begin{split}
\lim_{T \to 0} \partial_kI_A(0,k^{*})=\frac{3}{5} \sigma'(S_0),\\
\lim_{T \to 0} \partial_kI_E(0,k^{*})=\frac{1}{2} \sigma'(S_0).
  \end{split}
\end{align*}

\subsection{Approximations of implied volatility} \label{proxyNumerics}
In this last section we compare numerically the linear approximation of the implied volatility of an Asian and European call options given by equation \eqref{proxyPrice} with the classical Monte Carlo approximation. For the Asian case we consider the SABR model and for the European case we deal with the fractional Bergomi model.

\subsubsection{Asian call option under the SABR model}

We proceed with the following numerical experiment. We randomly sample the parameters as  $\sigma_0\sim U(0.2,0.8)$, $\alpha\sim U(0.3,1.5)$ and $\rho\sim U(-0.9,0.9)$, where $U$ stands for the uniform distribution. We fix $S_0=10$ and consider the following strikes $$K=( 9.97  ,  9.9743,  9.9786,  9.9829,  9.9872,  9.9915,  9.9958, 10.0001)$$ and maturities $T=(0.01, 0.1, 0.5, 1, 2)$. We consider the narrow OTM range around ATM region because due to the law of large numbers deep OTM quickly loose the value and prices become indistinguishable from 0. An Asian call option is priced using Monte Carlo with 100000 paths and discretization step 0.01 and the IV is  estimated using the same approach as discussed in Section 5.1. The approximation accuracy of the Monte Carlo for the IV is computed using the pointwise absolute relative error with respect to the 95$\%$ Monte Carlo confidence interval for each parameter. Then the median absolute percentatge error across 2000 random parameter combinations is computed in Table \ref{tab:ApproximationErrorMCSABR}. We then compare the Monte Carlo estimated IV with the approximation formula \eqref{proxyPrice}. For this, we compute in Table \ref{tab:sub_first0} the median absolute relative percentage error of the difference with the Monte Carlo prices computed across 2000 random parameter combinations.

\begin{table}[h!]
\begin{center}
\begin{tabular}{||c c c c c c c c c||}
\hline
T/K   & 9.9700  &  9.9743  &  9.9786  &  9.9829  &  9.9872  &  9.9915  &  9.9958  &  10.0001 \\ 
\hline\hline
0.01 & 0.94 & 0.92 & 0.91 & 0.91 &     0.91 &     0.92 &     0.93 &     0.94 \\
0.1 & 0.94 & 0.94 & 0.94 &  0.95 &  0.95 &     0.95 &     0.95 &     0.95 \\
0.5  & 1.02 & 1.02 &     1.02 &     1.02 &     1.02 &     1 &     1 & 1\\
1  & 1.1 &  1.1 &  1.1 &  1.1 & 1.1 &     1.1 &     1.1 &     1 \\
2 & 1.3 & 1.3 &  1.3 & 1.3 & 1.3 &  1.3 &  1.3 &   1.3 \\
\hline
\end{tabular}
\caption{Median absolute percentage error wrt the 95\% Monte Carlo confidence interval for an  Asian call IV under the SABR model.} \label{tab:ApproximationErrorMCSABR}
\end{center}
\end{table}

\begin{table}[h!]
\begin{center}
    \begin{tabular}{||c c c c c c c c c||}
    \hline
   T/K & 9.9700  &  9.9743  &  9.9786  &  9.9829  &  9.9872  &  9.9915  &  9.9958  &  10.0001 \\ 
    \hline\hline
0.01 & 0.24 & 0.23 & 0.22 & 0.23 & 0.23 & 0.24 &  0.24 & 0.72 \\
0.1 & 0.26 & 0.25 & 0.25 & 0.25 &  0.24 & 0.24 & 0.24 & 0.29 \\
0.5 & 0.89 & 0.89 &  0.87 & 0.87 & 0.87 & 0.87 & 0.86 & 0.81 \\
1 & 1.95 & 1.93 & 1.92 &  1.91 & 1.9 & 1.89 &  1.89 & 1.84 \\
2 & 4.04 & 4.04 & 4.03 & 4.01 & 4.02 & 4.05 & 4.02 & 4  \\
    \hline
    \end{tabular}
   \caption{Median absolute percentage error of the difference}\label{tab:sub_first0}
   \end{center}
\end{table}

We observe that our implied volatility approximation performs well in the ATM region for short dated options. We achieve implied volatility approximation accuracy comparable to the original Monte Carlo simulation. However, the quality of it impairs as soon as we move in the maturity ($T>0.5$) and moneyness dimension.

\subsubsection{European call option under the fractional Bergomi model}

In contrast to the Asian case, we implement Monte Carlo with 200000 paths and discretization step 0.02, and strikes $K=(9. ,9.1,  9.2, \dots, 10)$. We randomly sample the parameters of the model as $\sigma_0\sim U(0.2,0.8)$, $v \sim U(0.3,1.5)$ and $\rho\sim U(-0.9,0.9)$. We keep the values of $S_0$ and $T$ as in the SABR case. In order to investigate the influence of fractionalness of the volatility process we consider two values of $H=\{0.2,0.7\}$. 

We present the accuracy of the Monte Carlo pricing in Tables \ref{tab:ApproximationErrorMCERB02} and  \ref{tab:ApproximationErrorMCERB07} for $H=0.2$ and $H=0.7$, respectively. Then, the median absolute relative percentage error of the difference with the Monte Carlo prices computed across 2000 random parameter combinations and the approximation formula (\ref{proxyPrice}) is given in Tables \ref{tab:sub_firstE0} and \ref{tab:sub_firstE00}, respectively for $H=0.2$ and $H=0.7$.

\begin{table}[h!]
\begin{center}
\begin{tabular}{||c c c c c c c c c c c c||}
\hline
T/K  & 9.0  &   9.1  &   9.2  &   9.3  &   9.4  &   9.5  &   9.6  &   9.7  &   9.8  &  9.9  &  10.0 \\ 
\hline\hline
0.01     &  5.70 &  5.81 &  5.86 &  5.99 &  6.05 &  6.24 &  6.46 &  6.71 &  7.15 &  2.00 &  0.67 \\
0.10     &  7.10 &  7.49 &  7.74 &  7.72 &  7.87 &  6.76 &  4.68 &  2.37 &  1.13 &  0.74 &  0.68 \\
0.50     &  6.16 &  5.36 &  4.18 &  3.19 &  2.28 &  1.58 &  1.21 &  0.95 &  0.80 &  0.73 &  0.70 \\
1.00     &  3.31 &  2.71 &  2.15 &  1.69 &  1.42 &  1.17 &  0.97 &  0.85 &  0.78 &  0.74 &  0.72 \\
2.00     &  1.87 &  1.64 &  1.43 &  1.26 &  1.10 &  0.97 &  0.88 &  0.83 &  0.78 &  0.76 &  0.74 \\
\hline
\end{tabular}
\caption{Median absolute percentage error wrt the 95\% Monte Carlo confidence interval for the European call IV under the fractional Bergomi (H=0.2) model.} \label{tab:ApproximationErrorMCERB02}
\end{center}
\end{table}
\begin{table}[h!]
\begin{center}
\begin{tabular}{||c c c c c c c c c c c c||}
\hline
T/K  & 9.0  &   9.1  &   9.2  &   9.3  &   9.4  &   9.5  &   9.6  &   9.7  &   9.8  &  9.9  &  10 \\ 
\hline\hline
0.01     &  4.26 &  4.35 &  4.27 &  4.24 &  4.89 &  4.39 &  4.87 &  4.92 &  3.79 &  0.86 &  0.66 \\
0.10     &  5.88 &  5.99 &  6.68 &  6.41 &  5.67 &  4.25 &  2.30 &  1.03 &  0.71 &  0.64 &  0.66 \\
0.50     &  4.93 &  3.95 &  2.91 &  2.10 &  1.45 &  1.07 &  0.86 &  0.74 &  0.68 &  0.67 &  0.67 \\
1.00     &  2.76 &  2.23 &  1.74 &  1.37 &  1.14 &  0.95 &  0.83 &  0.75 &  0.72 &  0.70 &  0.70 \\
2.00     &  1.86 &  1.62 &  1.41 &  1.22 &  1.07 &  0.96 &  0.89 &  0.84 &  0.80 &  0.78 &  0.77 \\
\hline
\end{tabular}
\caption{Median absolute percentage error wrt the 95\% Monte Carlo confidence interval for the European call IV under the fractional Bergomi (H=0.7) model.} \label{tab:ApproximationErrorMCERB07}
\end{center}
\end{table}
 \begin{table}
\begin{center}
    \begin{tabular}{||c c c c c c c c c c c c||}
    \hline
    T/K   & 9.0  &   9.1  &   9.2  &   9.3  &   9.4  &   9.5  &   9.6  &   9.7  &   9.8  &  9.9  &  10.0 \\ 
    \hline\hline
0.01     &  79.81 &  78.02 &  75.87 &  72.65 &  68.73 &  63.17 &  54.73 &  41.38 &  15.11 &  0.50 &  0.12 \\
0.10     &  45.08 &  40.44 &  34.55 &  26.61 &  17.64 &  10.22 &   5.74 &   3.15 &   1.56 &  0.68 &  0.36 \\
0.50     &  18.00 &  15.85 &  13.73 &  11.79 &   9.59 &   7.63 &   5.95 &   4.34 &   2.98 &  2.01 &  1.64 \\
1.00     &  17.40 &  15.76 &  14.42 &  12.59 &  10.99 &   9.32 &   7.54 &   6.04 &   4.56 &  3.62 &  3.27 \\
2.00     &  18.97 &  17.86 &  16.57 &  15.07 &  13.58 &  11.97 &  10.36 &   8.97 &   7.64 &  6.76 &  6.31  \\
    \hline
    \end{tabular}
   \caption{Median absolute percentage error of the difference $H=0.2$}\label{tab:sub_firstE0}
   \end{center}
\end{table}
\begin{table}[h!]
\begin{center}
    \begin{tabular}{||c c c c c c c c c c c c||}
    \hline
   T/K  & 9.0  &   9.1  &   9.2  &   9.3  &   9.4  &   9.5  &   9.6  &   9.7  &   9.8  &  9.9  &  10 \\ 
    \hline\hline
0.01     &  76.52 &  74.74 &  72.48 &  70.59 &  67.12 &  61.92 &  54.90 &  42.06 &  17.64 &  1.09 &  0.64 \\
0.10     &  44.01 &  39.03 &  33.18 &  25.90 &  16.32 &   9.26 &   4.73 &   1.97 &   0.98 &  1.06 &  1.53 \\
0.50     &   9.72 &   7.38 &   5.62 &   4.08 &   2.77 &   2.04 &   1.76 &   1.78 &   2.17 &  2.75 &  2.99 \\
1.00     &   5.39 &   4.22 &   3.30 &   2.64 &   2.24 &   2.19 &   2.35 &   2.74 &   3.15 &  3.80 &  3.96 \\
2.00     &   3.37 &   2.90 &   2.86 &   2.83 &   2.95 &   3.28 &   3.64 &   3.98 &   4.82 &  5.15 &  5.25  \\
    \hline
    \end{tabular}
   \caption{Median absolute percentage error of the difference $H=0.7$}\label{tab:sub_firstE00}
   \end{center}
\end{table}

We observe that the error gets higher when we go further from the ATM region. We also observe that the quality of the linear approximation improves with the increase of the maturity. This is due to the fact that IV smile becomes flatter as the option maturity increases.

\appendix

\section{A primer on Malliavin Calculus} \label{MCintro}

We introduce the elementary notions of the
Malliavin calculus used in this paper (see \cite{Nualart2018}). Let us consider a standard Brownian motion $Z=(Z_t)_{t \in [0,T]}$ defined on a complete probability space $(\Omega, \mathcal{F}, \mathbb{P})$ and the corresponding filtration $\mathcal{F}_t$ generated by $Z_t$. Let ${\cal S}^Z$ be the set of random variables of the form
\begin{equation}
\label{eq:2.1}
F=f(Z(h_{1}),\ldots ,Z(h_{n})),  
\end{equation}
with $h_{1},\ldots ,h_{n}\in L^2([0,T])$, $Z(h_i)$ denotes the Wiener integral of the function $h_i$, for $i=1,..,n$, and $f\in C_{b}^{\infty }(\mathbb{R}^n) $ 
(i.e., $f$ and all its partial derivatives are bounded). Then the Malliavin 
derivative of $F$, $D^Z F$,  is defined
as the stochastic process given by 
\begin{equation*}
D_{s}^ZF=\sum_{j=1}^{n}{\frac{\partial f}{\partial x_{j}}}(Z(h_{1}),\ldots ,Z(h_{n})) h_j(s), \quad s\in [0,T].
\end{equation*}
This operator is closable from $L^{p}(\Omega )$ to $L^p(\Omega; L^2([0,T])$, for all $p \geq 1$, and we denote by ${\mathbb{D}}_{Z}^{1,p}$ the
closure of ${\cal S}^Z$ with respect to the norm
$$
||F||_{1,p}=\left( \E\left| F\right|
^{p}+\E||D^Z F||_{L^{2}([0,T])}^{p}\right) ^{1/p}.
$$
We also consider the iterated
derivatives \(D^{Z,n}\) for all integers  \(n > 1\) whose domains will be denoted by
\(\mathbb{D}^{n,p}_Z\), for all $p \geq 1$. We will use the notation $\mathbb{L}_Z^{n,p}:=L^p([0,T];{\mathbb{D}}_{Z}^{n,p})$.

%One of the main results in Malliavin calculus is the Clark-Ocone formula, see Theorem 6.1.1 in \cite{Nualart2018}.
%\begin{theorem} \label{co}
%Let $F\in \mathbb{D}_{Z}^{1,2} \cap L^2(\Omega, \mathcal{F}_T, \mathbb{P})$. Then $F$ admits the following representation:
%$$
%F=E(F)+\int_0^T E_t(D^Z_t(F))dZ_t.
%$$
%\end{theorem}

The following theorem is an extension of the classical It\^o's Lemma for non-anticipating processes, see Proposition 4.3.1 in \cite{Alos2021c}.
\begin{theorem}[Anticipating It\^o's Formula]
\label{aito}
Consider a process of the form $$X_t=X_0+\int_0^t u_sdZ_s+\int_0^t v_s ds,$$ where $X_0$ is an $\mathcal{F}_0$-measurable random variable and $u$ and $v$ are $\mathcal{F}_t$-adapted processes in $L^2([0,T] \times \Omega)$.
Consider also a process $Y_t = \int_t^T \theta_sds$, for some $\theta \in \mathbb{L}_Z^{1,2}$. Let $F : [0,T] \times \mathbb{R}^2 \rightarrow \mathbb{R}$ be a $C^{1,2}([0,T] \times \mathbb{R}^2)$ function such that there exists a positive constant $C$ such that, for all $t \in [0,T]$, $F$ and its derivatives evaluated in $(t,X_t,Y_t)$ are bounded by $C$. Then it follows that for all $t \in [0,T]$,
\begin{align*}
F(t,X_t,Y_t) &= F(0,X_0,Y_0)+\int_0^t \partial_sF(s,X_s,Y_s)ds+\int_0^t \partial_xF(s,X_s,Y_s)u_s dZ_s\\
&\qquad +\int_0^t \partial_xF(s,X_s,Y_s)v_s ds+\int_0^t \partial_yF(s,X_s,Y_s)dY_s\\
&\qquad+\int_0^t \partial_{xy}^2F(s,X_s,Y_s)u_sD^{-}Y_sds +\frac{1}{2}\int_0^t \partial_{xx}^2F(s,X_s,Y_s)u_s^2 ds,
\end{align*}
where $D^{-}Y_s=\int_s^TD^Z_s\theta_rdr$ and the integral $\int_0^t \partial_xF(s,X_s,Y_s)u_s dZ_s$ is defined in the Skorohod sense since the process $\partial_xF(s,X_s,Y_s)u_s$ is not adapted.
\end{theorem}
\section{The Price of an Asian Call Option under the Bachelier Model}\label{bc_asian_price}
In this section by prove the closed form formula for the price of an arithmetic Asian option under the constant volatility Bachelier model.
\begin{theorem} \label{AsianPrice}
Consider the model \textnormal{(\ref{B_Epm})} in the case of constant volatility $\sigma$. Then, the price of an arithmetic Asian call option for $t \in [0,T ]$ satisfies
$$
B_A(t,S_t,y_t,k,\sigma)=\left(S_t\frac{T-t}{T}+\frac{y_t}{T}-k\right)N(d_A(k,\sigma))+\left(\frac{\sigma(T-t)\sqrt{T-t}}{T\sqrt{3}}\right)n(d_A(k,\sigma)),
$$
where
$$
d_A(k,\sigma)=\frac{S_t\frac{T-t}{T}+\frac{y_t}{T}-k}{\left(\frac{\sigma(T-t)\sqrt{T-t}}{T\sqrt{3}}\right)},
$$
$T$  is the maturity, $S_t$ is the stock price at time $t$, $k$ is the strike price, and $y_t=\int_0^tS_udu$ is the state variable.
\end{theorem}

\begin{proof} Firstly, notice the following representation 
\begin{align*}
\begin{split}
  \frac{1}{T}\int_0^T S_t dt &= \frac{1}{T}\int_0^T \left(S_0+\int_0^t \sigma dW_u\right) dt = S_0+\frac{\sigma}{T}\int_0^T W_s ds = \\
   &= S_0+\frac{\sigma}{T}\int_0^t (T-s) dW_s +\frac{\sigma}{T}\int_t^T (T-s) dW_s.
\end{split}
\end{align*}
The last term on the right hand side is normally distributed with mean zero and variance equal to 
\begin{align*}
\begin{split}
  \mathbb{E}\left( \frac{\sigma}{T}\int_t^T (T-s) dW_s\right)^2 &= \frac{\sigma^2}{T^2}\int_t^T (T-s)^2 ds= \frac{\sigma ^2 (T-t)^3}{3 T^2}.
\end{split}
\end{align*}

Due to the risk neutral pricing argument we know that $$B_A(t,S_t,y_t,k,\sigma)=\mathbb{E}_t\left(\frac{1}{T}\int_0^T S_t dt-k\right)_+.$$ Therefore, using the above formulas this can be equivalently written as
\begin{align*}
\begin{split}
  \mathbb{E}_t\left(\frac{1}{T}\int_0^T S_t dt-k\right)_+ &= \mathbb{E}_t\left( S_t\frac{T-t}{T}+\frac{y_t}{T}-k +\frac{\sigma}{T}\int_t^T (T-s) dW_s\right)_+ = \\
   &= \frac{1}{\sqrt{2 \pi}}\int_{\frac{k-u_t}{z_t}}^{\infty}\left(u_t-k+z_t x\right)e^{-\frac{x^2}{2}}dx = \\
   &=\left(u_t-k\right)N\left(\frac{u_t-k}{z_t}\right)+z_t n\left(\frac{u_t-k}{z_t}\right),
\end{split}
\end{align*}
where $u_t = S_t\frac{T-t}{T}+\frac{y_t}{T}$ and $z_t = \frac{\sigma(T-t)\sqrt{T-t}}{T\sqrt{3}}$.

The last step allows us to complete the proof. \end{proof}
\appendix

\bibliography{literature}

\begin{thebibliography}{10}

\bibitem{Alos2006}
Elisa Al{\`{o}}s.
\newblock {A generalization of the Hull and White formula with applications to option pricing approximation}.
\newblock {\em Finance and Stochastics}, 10(3):353--365, 2006.

\bibitem{Alos2021c}
Elisa Al{\`{o}}s, David {Garc{\'{i}}a Lorite}, and Dariusz Gatarek.
\newblock {\em {Malliavin Calculus in Finance}}.
\newblock Chapman and Hall/CRC, 2021.

\bibitem{Alos2018}
Elisa Al{\`{o}}s, David Garc{\'{i}}a-Lorite, and Aitor~Muguruza Gonzalez.
\newblock {On Smile Properties of Volatility Derivatives: Understanding the VIX Skew}.
\newblock {\em SIAM Journal on Financial Mathematics}, 13(1):32--69, 2022.

\bibitem{Alos2017}
Elisa Al{\`{o}}s and Jorge~A. Le{\'{o}}n.
\newblock {On the curvature of the smile in stochastic volatility models}.
\newblock {\em SIAM Journal on Financial Mathematics}, 8:373--399, 2017.

\bibitem{Alos2007a}
Elisa Al{\`{o}}s, Jorge~A Le{\'{o}}n, and Josep Vives.
\newblock {On the short-time behavior of the implied volatility for jump-diffusion models with stochastic volatility}.
\newblock {\em Finance and Stochastics}, 11(4):571--589, 2007.

\bibitem{mp2022}
Elisa Al\`os, Eulalia Nualart, and Makar Pravosud.
\newblock {On the implied volatility of Asian options under stochastic volatility models}.
\newblock {\em Applied Mathematical Finance}, 30:249--274, 2023.

\bibitem{ASENS_1900_3_17__21_0}
Louis Bachelier.
\newblock {Th{\'{e}}orie de la sp{\'{e}}culation}.
\newblock {\em Annales scientifiques de l'{\'{E}}cole Normale Sup{\'{e}}rieure}, 3e s{\'{e}}rie,:21--86, 1900.

\bibitem{https://doi.org/10.1002/fut.22315}
Jaehyuk Choi, Minsuk Kwak, Chyng~Wen Tee, and Yumeng Wang.
\newblock {A Black-Scholes user's guide to the Bachelier model}.
\newblock {\em Journal of Futures Markets}, 42(5):959--980, 2022.

\bibitem{euch}
Omar El~Euch, Masaaki Fukasawa, Jim Gatheral, and Mathieu Rosenbaum.
\newblock {Short-Term At-the-Money Asymptotics under Stochastic Volatility Models}.
\newblock {\em SIAM Journal on Financial Mathematics}, 10(2):491--511, 2019.

\bibitem{Figueroa-Lopez2016}
Jos{\'{e}}~E. Figueroa-L{\'{o}}pez and Sveinn {\'{O}}lafsson.
\newblock {Short-term asymptotics for the implied volatility skew under a stochastic volatility model with L{\'{e}}vy jumps}.
\newblock {\em Finance and Stochastics}, 20(4):973--1020, 2016.

\bibitem{fouque}
Jean-Pierre Fouque, George Papanicolaou, and K.~Ronnie Sircar.
\newblock {From the implied volatility skew to a robust correction to Black-Scholes American option prices}.
\newblock {\em International Journal of Theoretical and Applied Finance}, 04(04):651--675, 2001.

\bibitem{doi:10.1080/14697688.2016.1197410}
Masaaki Fukasawa.
\newblock {Short-time at-the-money skew and rough fractional volatility}.
\newblock {\em Quantitative Finance}, 17(2):189--198, 2017.

\bibitem{Galeeva2022}
Roza Galeeva and Ehud Ronn.
\newblock {Oil futures volatility smiles in 2020: Why the bachelier smile is flatter}.
\newblock {\em Review of Derivatives Research}, 25(2):173--187, 2022.

\bibitem{HoGoodmanLaurieS2003}
Jeffrey Ho and Laurie Goodman.
\newblock {Interest Rates-Normal or Lognormal?}, 2003.

\bibitem{jaeckel}
Peter Jaeckel.
\newblock {Implied Normal Volatility}.
\newblock {\em Wilmott}, 2017:54--57, 2017.

\bibitem{Nualart2018}
David Nualart and Eulalia Nualart.
\newblock {\em {Introduction to Malliavin Calculus}}.
\newblock Cambridge University Press, 2018.

\bibitem{Pirjol2016}
Dan Pirjol and Lingjiong Zhu.
\newblock {Short maturity Asian options in local volatility models}.
\newblock {\em SIAM Journal on Financial Mathematics}, 7:947--992, 2016.

\bibitem{https://doi.org/10.1111/j.1467-9965.2007.00326.x}
Walter Schachermayer and Josef Teichmann.
\newblock {How close are the option pricing formulas of Bachelier and Black-Merton-Scholes?}
\newblock {\em Mathematical Finance}, 18(1):155--170, 2008.

\bibitem{Terakado2019}
Satoshi Terakado.
\newblock {On the Option Pricing Formula Based on the Bachelier Model}.
\newblock {\em SSRN Electronic Journal}, 2019.

\end{thebibliography}

\end{document}